\documentclass[11pt]{article}
\usepackage[utf8]{inputenc}
\usepackage{hyperref}
\usepackage{fullpage}
\usepackage{amssymb}
\usepackage{mathtools}
\usepackage{amsthm}
\usepackage{color,soul}

\newtheorem{theorem}{Theorem}
\newtheorem{lemma}[theorem]{Lemma}
\newtheorem{definition}[theorem]{Definition}

\newtheorem{fact}[theorem]{Fact}

\newtheorem{cor}[theorem]{Corollary}
\newtheorem{task}[theorem]{Task}
\newtheorem{assume}[theorem]{Assumption}

\newcommand {\dmax} [2] {\ensuremath{\mathrm{D}_{\max}\left(#1 \| #2\right)}}
\newcommand {\dmaxeps} [3] {\ensuremath{\mathrm{D}^{#3}_{\max}(#1 \| #2)}}
\newcommand {\F}{\ensuremath{\mathrm{F}}}
\newcommand {\Pur}{\ensuremath{\mathrm{Pur}}}
\newcommand {\eps}{\varepsilon}
\newcommand {\ball} [2] {\ensuremath{\mathcal{B}^{#1}{(#2)}}}
\newcommand {\relent} [2] {\ensuremath{\mathrm{D}{(#1 \| #2)}}}
\newcommand {\varrelent} [2] {\ensuremath{\mathrm{V}{(#1 \| #2)}}}
\newcommand {\mutinf}[2]{\ensuremath{\mathrm{I}(#1 : #2)}}
\newcommand {\relenteps} [3] {\ensuremath{\mathrm{D}^{#3}{(#1 \| #2)}}}

\newcommand {\Tr}{\ensuremath{\mathrm{Tr}}}
\newcommand {\cP}{\mathcal{P}}
\newcommand {\cE}{\mathcal{E}}
\newcommand {\cU}{\mathcal{U}}
\newcommand {\cH}{\mathcal{H}}
\newcommand {\cD}{\mathcal{D}}
\newcommand {\cF}{\mathcal{F}}
\newcommand {\cB}{\mathcal{B}}
\newcommand {\id}{\mathrm{I}}
\newcommand {\bL}{\textbf{L}}
\newcommand {\ellse}{q}

\newcommand {\suppress}[1]{}

\newcommand{\ketbra}[1]{|#1\rangle\langle#1|}

\newcommand{\ket}[1]{|#1 \rangle}

\begin{document}

\title{Quantifying resource in catalytic resource theory}

\author{
Anurag Anshu\footnote{Center for Quantum Technologies, National University of Singapore, Singapore. \texttt{a0109169@u.nus.edu}} \qquad
Min-Hsiu Hsieh\footnote{Centre for Quantum Software and Information, University of Technology Sydney, NSW 2007, Australia. \texttt{Min-Hsiu.Hsieh@uts.edu.au}} \qquad
Rahul Jain\footnote{Center for Quantum Technologies, National University of Singapore and MajuLab, UMI 3654, 
Singapore. \texttt{rahul@comp.nus.edu.sg}} \qquad 
}

\maketitle

\begin{abstract}
We consider a general resource theory that allows the use of free resource as a catalyst. We show that the amount of `resource' contained in a given state, in the asymptotic scenario, is equal to the regularized \emph{relative entropy of resource} of that state, which then yields a straightforward operational meaning to this quantity. Such an answer has been long sought for in any resource theory since the usefulness of a state in information-processing tasks is directly related to the amount of resource the state possesses in the beginning. While we need to place a few assumptions in our resource theoretical framework,  it is still general enough and includes quantum resource theory of entanglement, coherence, asymmetry, non-uniformity, purity, contextuality, stabilizer computation and the classical resource theory of randomness extraction as special cases. Since our resource theoretic framework includes entanglement theory, our result also implies that the amount of noise one has to inject locally in order to erase all entanglement contained in an entangled state is equal to the regularized relative entropy of entanglement, resolving an open question posted in [\emph{Groisman et al., Phys. Rev. A. 72: 032317, 2005}].  On the way to prove the main result, we also quantify the amount of resource contained in a state in the one-shot setting (where one only has a single copy of the state), in terms of the smooth max-relative entropy. Our one-shot result employs a recently developed technique of convex-split lemma.
\end{abstract}

\section{Introduction}

A question that is commonly asked in Physics is how a certain property of a physical system can be used to achieve useful tasks, and how to quantify the amount of such a property in a meaningful way. Various areas in Physics have developed methodology to answer these questions relevant to their own areas; however, many of the approaches are difficult, sometimes even impossible, to employ outside of their respective fields. Quantum resource theory then emerges when  quantum information theory is found to provide a unified platform for characterizing the resource \cite{Devetak2008, Datta2011} because in a nutshell, they can all be viewed as interconversion  of different system states with system-dependent constraints. Since then, individual resource theory has been able to characterize targeted properties in some information-processing tasks using entropic quantities defined on system states.  

The core of a resource theory is built upon two main system-dependent requirements for the resources and allowed operations; namely, (i) the existence of a set of states  that are free and those not in the set are expensive; and (ii) the allowed operations are those that map the set of free states to itself. Various resource theories have been developed in the past decade. The most notable example is the resource theory of entanglement \cite{Horodecki2009}, where the set of free states corresponds to the collection of separable states and the allowed free operations are the local quantum operations and classical communication (LOCC). These two classes of states and operations have attracted stand-alone interests besides resource theory \cite{Chitambar14, Chitambar-2014b}. Under this resource framework, one can then ask the amount of valuable resource (cf.~quantum entanglement) possessed by an entangled state relative to the set of free states (cf.~separable states). This question motivates the scenario of injecting noise locally to the system in order to destroy the quantum entanglement, i.e., the \emph{randomness cost}. A complete characterization of this question has remained open; though, gapped upper and lower bounds have been provided in the asymptotic i.i.d.~setting in Ref.~\cite{Groisman05}. It is worthwhile to mention that investigation of a variant of the above setting, where the local noise is used to destroy the total correlation in an entangled state, relates the minimal randomness cost to the quantum mutual information \cite{Groisman05}. This seminal result gives the first operational meaning to this entropic quantity and advances significantly  our understanding of entanglement theory. Being able to answer the optimal randomness cost to bring entangled states to separable states thus bears equivalent significance, if not more important, since the existence of entanglement is believed to make quantum systems superior to their classical counterparts  and the amount of entanglement is generally linked to its information-processing power \cite{Shor99}. Likewise, this crucial question of the amount of valuable resource possessed by a state relative to its free resource is then adhesive to every resource theory, be it quantum coherence, thermodynamics, etc. 

One can view the above erasing framework as a restricted model of state transformation, where the final state belongs to the free resource \cite{BrandaoGour15, PhysRevA.95.062314}. In the general model of state transformation, it has been demonstrated that not every state transformation is possible, and could become possible if a catalyst is involved \cite{Feng05}. Even if the transformation process is possible in the beginning, adding catalyst would likely make the process much more efficient. Therefore, individual resource theories have included catalyst in their formalism \cite{GOUR20151, Brandao2015}. 

We thus consider a general resource theory framework that allows free resource being used as a catalyst. We demonstrate that this framework includes previous major resource theories of (i) entanglement \cite{Horodecki2009, Groisman05, BrandaoPlenio08, ChitambarHsieh16}, (ii) coherence \cite{BaumgratzCP14, WinterYang16, ChitambarHsieh16, Streltsov16}, (iii) thermodynamics \cite{BrandaoHORS13, BrandaoHNOW15, HorodeckiO15, FaistDJR15, GourMNSH15, VarunG15}, (iv) non-uniformity \cite{GOUR20151}, (v) purity \cite{PhysRevA.67.062104}, (vi) randomness extractors \cite{Trevisan01, TrevisanV00}, (vii) contextuality \cite{HorodeckiGJKL15}, (viii) asymmetry \cite{Wakakuwa17} and (ix) stabilizer computation \cite{VeitchMGE14} as special instances. Treating resource theory in a general framework enjoys the major advantage that the resource conversion tasks done for one particular resource can already be adapted to other resources almost trivially. Hence, it has started to attract more attention \cite{BrandaoGour15, PhysRevA.95.062314}. We establish that the entropic quantity, the \emph{regularized relative entropy of resource} \cite{BrandaoGour15}, characterizes the amount of useful resource in a given state relative to its free resource in the asymptotic i.i.d.~setting. In other words, if one has infinitely many copies of the same state, the randomness cost (per copy of the state) for erasing the resource contained in the given state and bringing it to the closest free state is equal to this entropic quantity. This yields a crucial operational meaning to this entropic quantity. Since the resource theory of quantum entanglement is a special case of our framework, our result directly provides the missing answer\footnote{We remark that an independent work for this answer can be found in Ref.~\cite{Berta2017u}} to the minimal randomness cost for erasing the amount of entanglement in an entangled state to a separable state, a question first posted in Ref.~\cite{Groisman05}. 

We also obtain matching upper and lower bounds on the randomness cost if one only has a single copy of a given state, i.e., the one-shot scenario. Our one-shot bounds are given in terms of smooth max-relative entropy; hence our result also provides a new operational meaning to this quantity in the resource theoretic framework. We emphasize that being able to obtain matching one-shot bounds in such a general resource theory framework is rare \cite{BrandaoGour15}, and is due to the technical tool: the convex split lemma \cite{AnshuDJ14}. This again provides another excellent example that quantum information theory helps in the understanding of quantum physics. 


The paper is organized as follows. Section~\ref{sec_main} contains  assumptions required for our catalytic resource framework. We then define our information-processing task of quantifying the amount of resource and present our main theorem. We also discuss several individual resource theories as special cases of our framework. Section~\ref{sec_Pre} contains definition of entropic quantities and useful lemmas required in the proofs. 
The proofs of our main result are given in Sections~\ref{sec_oneshot}--\ref{sec_asymp}. We conclude in Section~\ref{sec_cld}.

\section{Catalytic Resource Framework and Main Results}\label{sec_main}

The question that we will answer in this general catalytic resource theory framework is how much resource is contained in a given resource state $\rho$  (relative to its free resource). This question is crucial, as this amount directly relates to the power the state $\rho$ possesses in achieving information-processing tasks. 

We first present assumptions that we have to make in our resource theoretic framework. It is apparent that if no limit is set on the allowed operations and free resources, then it is almost impossible to obtain useful characterization, as also noted in Ref.~\cite{BrandaoGour15}. Thus, we need a sufficient (yet small) number of assumptions to keep the theory interesting. 

Our framework for both single-partite and multi-partite settings is discussed below, with the assumptions on free resources and allowed operations. Our assumptions capture the requirement we impose on an experimenter (or experimenters in the multi-partite setting) who wants to quantify a useful property of his/her physical system. For this, we start with the natural experimental framework, where the experimenter has decided upon the global Hilbert space and its tensor decomposition into several Hilbert spaces (for example, as `physically separated' Hilbert spaces). Given this framework, we impose the conditions that are to be satisfied by the set of free resources and the allowed unitary operations. It is important to note that the conditions are only for a given framework; we do not restrict the \textit{choice} of the global Hilbert space and its tensor decomposition by the experimenter. 

For the notations and definitions used in this section, please refer to Section \ref{sec_Pre}.

\subsection{Single-partite case}

As discussed above, an experimenter starts with a collection of quantum registers and performs operations on them. The joint state lives on a global Hilbert space $\cH$. This Hilbert space possesses an ordered decomposition into $r$ systems $\cH = \cH_{1}\otimes \cH_{2}\otimes \ldots \cH_{r}$. Given the Hilbert space and its ordered decomposition (the order of the decomposition may matter in a given physical scenario, hence we take it into account), we define the set of free states and allowed operations. The definition is analogous to the postulates given in Ref.~\cite{BrandaoGour15}.

\begin{definition}[Single party resource theory]
\label{def:generictransform}
Fix an integer $r\geq 1$. Given $\cH= \cH_{1}\otimes \cH_{2}\otimes \ldots \cH_{r}$ and registers $M_1,M_2,\ldots M_r$ where $M_i$ corresponds to $\cH_i$, the set of free resources and free operations are as follows.

\begin{enumerate}
\item For every non-empty ordered set $\bL = \{s_1,s_2, \ldots s_{|\bL|}\}\subseteq \{1,2,\ldots r\}$ (where $\{s_1,s_2,\ldots s_{|L|}\}$ are increasing sequence of integers), let $\cH_{\bL} = \cH_{s_1}\otimes \cH_{s_2} \otimes \ldots \cH_{s_{|L|}}$ and $\cF_\bL \subseteq \cD(\cH_{\bL})$ satisfy the following properties. For brevity, we set $M_{\bL}:= M_{s_1} M_{s_2}\ldots  M_{s_{|L|}}$
\begin{enumerate}
\item $\cF_\bL$ is a convex and closed set. 
\item For any $\bL' \subseteq \bL$, $\cF_{\bL'}\otimes \cF_{\bL\setminus \bL'} \subseteq \cF_\bL$.
\begin{itemize}
\item This assumption states that if two quantum states are free resources, then their tensor product is a free resource as well.
\end{itemize}
\item For any $\sigma_{M_{\bL}} \in \cF_{\bL}$ and $\bL' \subseteq \bL$,  $\Tr_{M_{\bL'}}(\sigma_{M_{\bL}}) \in \cF_{\bL\setminus \bL'}$.
\begin{itemize}
\item This assumption states that if a quantum state on more than one registers is a free resource, then we obtain a free resource by partial trace over a subset of these registers.
\end{itemize}
\item It holds that $\frac{\id_{M_r}}{|M_r|}\in \cF_{\{r\}}$, where $|M_r|$ is the dimension of register $M_r$.
\begin{itemize}
\item Under this assumption, we require the maximally mixed quantum state to be a free resource on the last register. 
\end{itemize}
\end{enumerate}
\item The complete set of free resources is $\cF := \cup_{\bL}\cF_\bL$.

\item For a fixed $\bL$, the free operations $\cU_\bL$ are defined as the set of all unitaries $U: \cH_\bL\rightarrow \cH_\bL$ that satisfy
\begin{enumerate}
\item $\sigma\in \cF_\bL \implies U\sigma U^{\dagger}\in \cF_\bL$. 
\item $U\in \cU_\bL $ if and only if $U^{\dagger}\in \cU_\bL $.
\end{enumerate}
\item The set of all free operations are $\cU:= \cup_\bL \cU_\bL$.

\end{enumerate}

\end{definition}

\noindent We remark that these assumptions we required are natural and mild \cite{BrandaoGour15}.  Above, we also observe that the set $\cU_\bL$ (for every $\bL$) forms a group.

Now we are in a position to formally define our task, that we call an $(\eps, \log|J|)$-transformation of $\rho_{M}$ to $\cF$.

\begin{task}
\label{catalytictransform}
Let $\eps >0$, $r\geq 1$ be an integer and fix a Hilbert space $\cH$ with decomposition $\cH_{1}\otimes \cH_{2}\otimes \ldots \otimes\cH_{r}$, which is chosen by the experimenter. Let the register $M$ correspond to $\cH_M =\cH_{1}$, register $J$ correspond to $\cH_J = \cH_{r}$ and register $E$ correspond to $\cH_E=\cH_{2}\otimes \ldots \cH_{r-1}$. An experimenter holds a quantum state $\rho_{M}$. Using a state $\mu_{EJ} \in \cF$, she applies a unitary $U \in \cU$  to obtain a joint quantum state $\Theta_{MEJ}$:
\[
\Theta_{MEJ} = U (\rho_M\otimes\mu_{EJ}) U^\dagger.
\] 
It is required that there exists a $\sigma_{ME}\in \cF$ such that $\Pur(\Theta_{ME}, \sigma_{ME})\leq \eps$, where the chosen distance measure is defined in Equation (\ref{eq_dist}). The number of discarded qubits is $\log |J|$.  
\end{task}

In above task, no restriction is made on the transformation of the free resource state $\mu_{EJ}$. But in many cases, it is desirable that this quantum state be returned close to its original form and hence act as a catalyst. For example, the framework of catalytic decoupling \cite{MajenzBDRC17} studies such a scenario. Our achievability result shall belong to such class of transformations. Hence we have the following definition.

\begin{definition}
\label{actualcatalytic}
Task \ref{catalytictransform} is said to be a $(\eps, \log|J|)$-\textit{catalytic transformation} of $\rho_M$ if $\mu_{EJ} = \mu_E\otimes \mu_J$ and $\sigma_{ME} = \sigma_M\otimes \mu_E$ for some $\sigma_M\in \cF$. 
\end{definition}

Task \ref{catalytictransform} is motivated by the work \cite{Groisman05}, which considered the problem of transforming a bipartite quantum state $\rho_{AB}^{\otimes n}$ into a product state $\rho_{A^n}\otimes \rho_{B^n}$ with the aid of shared randomness and local unitaries. It was shown that the number of bits of randomness required (in other words, the randomness cost) is $\approx n\cdot \mutinf{A}{B}_{\rho}$. Task \ref{catalytictransform} is along the lines of the framework considered in above work, but with an additional freedom of allowing the use of additional free resources that can aid in the transformation of the desired quantum state. Moreover, when the register $J$ is classical (in a more precise sense to be discussed below), we can interpret $\log|J|$ to be the randomness cost of the protocol. In Definition \ref{actualcatalytic}, we have provided further restriction that the free resource (expect for the randomness used) be returned with small error.

We provide a near optimal characterization of the randomness cost of Task \ref{catalytictransform} in Theorem \ref{thm_oneRER} below. Our achievability protocol requires a further assumption related to the last register (and its connection to rest of the registers), as this register serves as the source of `classical randomness'. This assumption is as follows. 

\begin{assume}
\label{assumeperm}
Invoke the notation in Definition \ref{def:generictransform}. Then there exists a canonical basis $\cB = \{\ket{1}, \ket{2}, \ldots \ket{\ell_r}\}$ on $\cH_{r}$ (where $\ell_r$ is the dimension of $\cH_r$) such that 
$$\cF = \text{conv}\left\{\frac{1}{\ell_r}\sum_j U_j\sigma U^{\dagger}_j\otimes \ketbra{j}_{M_r}: \sigma\in \cF_{\{1,2,\ldots r-1\}}, U_j \in \cU_{\{1,2,\ldots r-1\}}\right\},$$
and the experimeter only performs an operation from the set
$$\left\{\sum_{j=1}^{\ell_r} U_j \otimes \ketbra{j}_{M_r}: U_j \in \cU_{\{1,2,\ldots r-1\}}\right\}.$$

Moreover, suppose it holds that $\cF_{\{j\}} = \cF_{\{1\}}$ for all $j<r$. Then $\sigma_{M_1M_2\ldots M_r}\in \cF$  implies $\sigma_{M_jM_1\ldots M_{j-1}M_{j+1}\ldots M_r}\in \cF$ for all $j<r$.
\end{assume}

Under this assumption, all the states $\mu_{EJ}\in \cF$ in Task \ref{catalytictransform} are classical-quantum states, with register $J$ being classical and $\mu_J$ being diagonal in the canonical basis. Thus, the assumption allows us to interpret $\log|J|$ as the randomness cost of the protocol. We have the following theorem.

\begin{theorem} \label{thm_oneRER}
Fix $\epsilon,\delta>0$, and a quantum state $\rho_{M}$. 
\begin{itemize}
\item Achievability: Suppose Assumption \ref{assumeperm} holds. Then there exists an $(\eps+\delta, k + 2\log\frac{1}{\delta})$-catalytic transformation of $\rho_{M}$ to $\cF$ , where $k := \min_{\sigma_{M}\in \cF}\dmaxeps{\rho_{M}}{\sigma_{M}}{\eps}$.
\item Converse: For every $(\eps, \log|J|)$- transformation of $\rho_M$ to $\cF$, it holds that 
$$\log |J| \geq \min_{\sigma_M\in \cF}\dmaxeps{\rho_M}{\sigma_M}{\eps},$$ if Assumption \ref{assumeperm} is true. Otherwise, 
$$\log |J| \geq \frac{1}{2}\min_{\sigma_M\in \cF}\dmaxeps{\rho_M}{\sigma_M}{\eps}.$$ 
\end{itemize}

\end{theorem}
\noindent The proof of this theorem is given in Section \ref{sec_oneshot}. Now we consider the asymptotic i.i.d. properties of our results.

\begin{definition}
\label{def:asymptotic}
Suppose Assumption \ref{assumeperm} holds. We say that the asymptotic randomness \textit{rate} of catalytic transformation of $\rho_M$ is $R$, if for every $\eps>0$, there exists an $n_0(\eps)$ such that for all $n\geq n_0(\eps)$, there exists a $(\eps, nR)$-catalytic transformation of $\rho^{\otimes n}_M$ to $\cF$.
Define the following relative entropy of resource:
\begin{equation*}
E(\rho) = \inf_{\sigma_M\in \cF} D(\rho_M\|\sigma_M)
\end{equation*}
and the regularized relative entropy of resource:
\[
E^\infty (\rho) = \lim_{n\to \infty} \frac{1}{n} E(\rho_M^{\otimes n}). 
\]
Let $\cF_n$ denote the set of free resources for the register $M^n:= M\times M\times\ldots M$. Define the constant 
$$C(\cF):= \lim_{n\rightarrow \infty}\frac{1}{n}\inf_{\tau\in \cF_n}\|\log \tau\|_\infty.$$
\end{definition}
Using Theorem~\ref{thm_oneRER}, we obtain the following.

\begin{theorem}\label{coro_RER}
Suppose the resource theory $\cF$ is such that $C(\cF) < \infty$. Then for a quantum state $\rho_M$, the asymptotic randomness rate of catalytic transformation of $\rho$ is given by $E^\infty (\rho)$.
\end{theorem}
It completely characterizes the per copy randomness requirement in the asymptotic i.i.d. setting.
The proof is given in Section~\ref{sec_asymp}. Moreover, we show that the number of qubits of catalytic free resource grows only polynomially in $n$ (Theorem \ref{theo:asymptotictransform}).

\subsection{Multi-partite case}

Our formalism also extends to the multi-partite case. We consider the case of $t$ parties and introduce some minor modifications to Definition \ref{def:generictransform} and Assumption \ref{assumeperm}. Let $\cH$ be the global Hilbert space in the multiparty setting. We start with a decomposition of $\cH = \cH_{1} \otimes \cH_2 \ldots \otimes \cH_r$, where $\cH_i$ is a Hilbert space shared by all the $t$ parties.

\begin{definition}[Multi-partite resource theory]
\label{def:multigenerictransform}
Fix an integer $r\geq 1$. Given $$\cH = \cH_1 \otimes \cH_2 \ldots \otimes \cH_r$$ and registers $M^1_1M^2_1\ldots M^t_1$, $M^1_2M^2_2\ldots M^t_2$, $\ldots, M^1_rM^2_r\ldots M^t_r$, where $M^1_iM^2_i\ldots M^t_i$ corresponds to $\cH_i$ and $M^j_i$ is held by $j$-th party, the set of free resources and free operations are as follows.

\begin{enumerate}
\item For every non empty ordered set $$\bL = \{s_1,s_2,\ldots s_{|\bL|}\}\subseteq \{1,2,\ldots r\}$$ let $\cF_\bL \subseteq \cD(\cH_{s_1}\otimes \cH_{s_2} \otimes \ldots \otimes\cH_{s_{\bL}})$ satisfy the following properties. For brevity, we let $M_{\bL}$ be the register corresponding to set $\bL$.
\begin{enumerate}
\item $\cF_\bL$ is a convex and closed set. 
\item For any $\bL' \subseteq \bL$, $\cF_{\bL'}\otimes \cF_{\bL\setminus \bL'} \subseteq \cF_\bL$.
\item For any $\sigma_{M_{\bL}} \in \cF_{\bL}$ and $\bL'  \subseteq \bL$,  $\Tr_{M_{\bL'}}(\sigma_{M_{\bL}}) \in \cF_{\bL\setminus \bL'}$.

\end{enumerate}
\item The complete set of free resources is $\cF := \cup_{\bL}\cF_\bL$.

\item For a fixed $\bL$, the free operations $\cU_\bL$ are defined as the set of all unitaries $U$ that satisfy
\begin{enumerate}
\item $\sigma\in \cF_\bL \implies U\sigma U^{\dagger}\in \cF_\bL$. 
\item $U\in \cU_\bL $ if and only if $U^{\dagger}\in \cU_\bL $.
\end{enumerate}
\item The set of all free operations are $\cU:= \cup_\bL \cU_\bL$.
\end{enumerate}

\end{definition}

In an analogous fashion to Task \ref{catalytictransform}, we define the multi-partite task.

\begin{task}[An $(\eps, \log|J|, t)$-transformation of $\rho_{M^1M^2\ldots M^t}$ to $\cF$]
\label{multicatalytictransform}
Let $\eps >0$, $r\geq 1$ be an integer and fix a Hilbert space $\cH$ with decomposition $\cH_{1} \otimes \cH_2 \ldots \otimes \cH_r$, which is chosen by the experimenters. Let the register $M^1M^2\ldots M^t$ correspond to $\cH_{M^1M^2\ldots M^t} = \cH_{1}$, register $J_1J_2\ldots J_t$ correspond to $\cH_{J_1J_2\ldots J_t} = \cH_{r}$ (with register $J_i$ held by $i$-th party) and register $E$ correspond to $\cH_E = \cH_{2} \otimes \cH_{3} \ldots \otimes \cH_{r-1}$. The experimenters share a quantum state $\rho_{M^1M^2\ldots M^t}$. Using a state $\mu_{EJ_1J_2\ldots J_t} \in \cF$, they apply the unitary $U \in \cU$  to obtain a joint quantum state $\Theta_{M^1M^2\ldots M^tEJ_1J_2\ldots J_t}$:
\[
\Theta_{M^1M^2\ldots M^tEJ_1J_2\ldots J_t} = U (\rho_{M^1M^2\ldots M^t}\otimes\mu_{EJ_1J_2\ldots J_t}) U^\dagger.
\] 
It is required that there exists a $\sigma_{M^1M^2\ldots M^tE}\in \cF$ such that $\Pur(\Theta_{M^1M^2\ldots M^tE}, \sigma_{M^1M^2\ldots M^tE})\leq \eps$, where the chosen distant measure is defined in Equation (\ref{eq_dist}). 
\end{task}

We can similarly define a catalytic transformation in this framework.

\begin{definition}
\label{actualmulticatalytic}
Task \ref{multicatalytictransform} is said to be a $(\eps, \log|J|, t)$-\textit{catalytic transformation} of $\rho_{M^1M^2\ldots M^t}$ if $\mu_{EJ_1J_2\ldots J_t} = \mu_E\otimes  \mu_{J_1J_2\ldots J_t}$ and $\sigma_{M^1M^2\ldots M^tE} = \sigma_{M^1M^2\ldots M^t}\otimes \mu_E$ for some $\sigma_{M^1M^2\ldots M^t}\in \cF$. 
\end{definition}

We will also need the following assumption for achievability protocol.

\begin{assume}
\label{assumepermmul}
Invoke the notation from Definition \ref{def:multigenerictransform}. There exists a canonical basis $\cB_j = \{\ket{1}, \ket{2}, \ldots \ket{\ell_r}\}$ on $\cH^j_{r}$ (where the dimension $\ell_r$ of $\cH^j_{r}$ is independent of $j$) such that 
$$\cF= \text{conv}\left\{\frac{1}{\ell_r}\sum_{k=1}^{\ell_r} U_k\sigma U_k^{\dagger}\otimes\ketbra{k}_{M^1_r}\otimes \ketbra{k}_{M^2_r}\otimes\ldots \ketbra{k}_{M^t_r}: U_k \in \cU_{\{1,2,\ldots r-1\}}\right\}$$ 
and experimenter only performs an operation from the set $$\left\{\sum_k U_k \otimes \ketbra{k}_{M^1_r}\otimes \ketbra{k}_{M^2_r}\otimes \ldots \ketbra{k}_{M^t_r}:  U_k \in \cU_{\{1,2,\ldots r-1\}} \right\}.$$
For brevity, define $$\id_{\ell_r,t} := \frac{1}{\ell_r}\sum_{k=1}^{\ell_r}\ketbra{k}_{M^1_r}\otimes \ketbra{k}_{M^2_r}\otimes\ldots \ketbra{k}_{M^t_r}.$$
Moreover, suppose it holds that $\cF_{\{j\}} = \cF_{\{1\}}$ for all $j<r$. Then $$\sigma_{M^1_1M^2_1\ldots M^t_1, M^1_2M^2_2\ldots M^t_2,\ldots M^1_rM^2_r\ldots M^t_r}\in \cF$$  implies $$\sigma_{M^1_jM^2_j\ldots M^t_j, M^1_2M^2_2\ldots M^t_2,\ldots M^1_{j-1}M^2_{j-1}\ldots M^t_{j-1}, M^1_1M^2_1\ldots M^t_1, M^1_{j+1}M^2_{j+1}\ldots M^t_{j+1}, \ldots M^1_rM^2_r\ldots M^t_r}\in \cF$$ for all $j<r$.

\end{assume}

We are now in a position to prove the following theorem.

\begin{theorem} \label{thm_multiRER}
Fix $\epsilon,\delta>0$, and a quantum state $\rho_{M^1M^2\ldots M^t}$. 
\begin{itemize}
\item Achievability: Suppose Assumption \ref{assumepermmul} holds. There exists an $(\eps+\delta, k + 2\log\frac{1}{\delta}, t)$-catalytic transformation of $\rho_{M^1M^2\ldots M^t}$ to $\cF$, where $k := \min_{\sigma_{M^1M^2\ldots M^t}\in \cF}\dmaxeps{\rho_{M^1M^2\ldots M^t}}{\sigma_{M^1M^2\ldots M^t}}{\eps}$.
\item Converse: For every $(\eps, \log|J|, t)$-transformation of $\rho_{M^1M^2\ldots M^t}$ to $\cF$, it holds that 
$$\log |J| \geq \min_{\sigma_{M^1M^2\ldots M^t}\in \cF}\dmaxeps{\rho_{M^1M^2\ldots M^t}}{\sigma_{M^1M^2\ldots M^t}}{\eps},$$ if Assumption \ref{assumepermmul} is true.
\end{itemize}

\end{theorem}

The proof of this theorem follows similar to the proof of Theorem \ref{thm_oneRER} and is given in Corollaries \ref{cor:multipartite} and \ref{cor:converse} in Section \ref{sec_multioneshot}.

\subsection{Known resource theories}
To conclude this section, we show that our general resource framework includes at least the following individual resource theory. 

\begin{enumerate}
\item{Resource theory of entanglement \cite{Horodecki2009, Groisman05, BrandaoPlenio08, ChitambarHsieh16}:} In this resource theory, the set of free states, $\cF$, is the collection of separable states. The perfectly correlated state $\id_{\ell, 2}$ is a separable quantum state for all $\ell\geq 1$, and hence is a free resource. The formalism also extends to multi-partite entanglement, where the free resources are convex combination of product quantum states.

\item{Resource theory of coherence \cite{BaumgratzCP14, WinterYang16, ChitambarHsieh16, Streltsov16}:} In this resource theory, the set of free states, $\cF$, is the collection of diagonal states in a pre-determined basis. This is captured by Definition \ref{def:generictransform}, where we choose a basis on each Hilbert space $\cH_i$ and take the free resources $\cF_{\{i\}}$ to be the set of diagonal quantum states. The properties such as being closed under partial trace and tensor product are easily seen to be satisfied. Moreover, maximally mixed state is diagonal in any basis, and hence belongs to the free resources.

\item{Resource theory of asymmetry \cite{Wakakuwa17}:} In this resource theory, the set of free resources are the states that are invariant under some group transformation. To construct $\cF_{\{\ell\}}$ for some $\ell\geq 2$, one fixes a group $G$ with a unitary representation $\{U_g\}_{g\in G}$ in dimension $\ell$. The set $\cF_{\{\ell\}}$ is the collection of states $\sigma$ that satisfy $U_g\sigma U^{\dagger}_g = \sigma$. The resource theory can then be constructed on tensor product of such Hilbert spaces (see Section $2.B$ in \cite{Wakakuwa17}).  To verify the partial trace condition (Item $1(c)$), let $\sigma_{M_1M_2}$ be such that $(U_g\otimes V_{g'})\sigma_{M_1M_2} (U^{\dagger}_g\otimes V^{\dagger}_{g'}) = \sigma_{M_1M_2}$. Then 
\begin{eqnarray*}
\sigma_{M_1}&=& \Tr_{M_2}(\sigma_{M_1M_2})=\Tr_{M_2}\left(U_g\otimes V_{g'})\sigma_{M_1M_2} (U^{\dagger}_g\otimes V^{\dagger}_{g'})\right) \\ &=& \Tr_{M_2}\left(U_g\otimes V^{\dagger}_{g'}V_{g'})\sigma_{M_1M_2} (U^{\dagger}_g\otimes \id)\right) = \Tr_{M_2}\left(U_g\otimes \id)\sigma_{M_1M_2} (U^{\dagger}_g\otimes \id)\right)\\ &=& U_g\sigma_{M_1}U^{\dagger}_g.
\end{eqnarray*}
Symmetry under permutation of registers of identical dimension can also be verified in similar way, as $\sigma_{M_1M_2}$ must be invariant under $U_g\otimes U_{g'}$ for all $g,g'$.  

\item{Resource theory of nonuniformality \cite{GOUR20151} and purity  \cite{PhysRevA.67.062104}:} In this resource theory, the only  free state is the completely mixed state. This is easily captured by Definition \ref{def:generictransform}. This is equivalent to the formalism of randomness extractors \cite{Trevisan01, TrevisanV00}.  

\item{Resource theory of Thermodynamics \cite{BrandaoHORS13, BrandaoHNOW15, HorodeckiO15, FaistDJR15, GourMNSH15, VarunG15}:} In this resource theory, the free states are Gibbs quantum states, that are states of the form $\rho_\beta(H)=\frac{e^{-\beta H}}{\Tr(e^{-\beta H})}$, for an arbitrary Hamiltonian $H > 0$. To capture it in Definition \ref{def:generictransform}, we assign a Gibbs state as the only free resources in the set $\cF_{\{i\}}$ (for all $i$). Thus, for a system with $r$ registers, the free resource is $\otimes_{i=1}^r \rho_\beta(H_i)$. The condition that $U \otimes_{i=1}^r \rho_\beta(H_i) U^{\dagger} = \otimes_{i=1}^r \rho_\beta(H_i)$ is equivalent to the condition $[U,\sum_i H_i]=0$.  Moreover, this also implies $[U^{\dagger},\sum_i H_i]=0$. Finally, the maximally mixed state is a Gibbs quantum state with Hamiltonian $H=0$. This establishes consistency with Definition \ref{def:generictransform}. 

Theorem \ref{coro_RER} applies as long as $C(\cF) < \infty$. For a Gibbs state $\rho_{\beta}(H)$, it holds that $$\|\log \rho_{\beta}(H)\|_\infty \leq \|\beta H\|_\infty + \log\Tr(e^{-\beta H}) \leq \|\beta H\|_\infty + \log d,$$ where $d$ is the support size of the Gibbs state. Thus, Theorem \ref{coro_RER} applies for bounded $\|\beta H\|_\infty$.

\item{Resource theory of contextuality \cite{HorodeckiGJKL15}:}  In this resource theory, the free resources are the set of conditional probability distributions that can be described `classically' (a notion that is made more precise in \cite{HorodeckiGJKL15}). All the conditional probability distributions in this theory satisfy a consistency condition (which can be viewed as an analogue of the non-signaling condition in the case of non-locality). It can be observed that the set of free resources satisfies the conditions in Definition \ref{def:generictransform}, being convex and closed under permutations and tensor product. Uniform probability distribution (or the maximally mixed state) also belongs to the free resource.

\item{Resource theory of stabilizer computation \cite{VeitchMGE14}:} In this resource theory, the free resources are the convex hull of the set of all pure states that can be generated by the action of a Clifford unitary on a standard state (for example, $\ket{0}$). Thus, maximally mixed state belongs to this set. This set is convex and closed. The free resources can be extended to the tensor product of registers along the lines as discussed in \cite{Wakakuwa17}. 

\end{enumerate}

\section{Preliminaries}\label{sec_Pre}


Let $\cH_A$ be the Hilbert space associated to a register $A$. Let $|A|$ denote the dimension of $\cH_A$. Let $\mathcal{D}(\cH_A)$ be the set of all normalized quantum states acting on $\cH_A$. For two subsets $\cF_A \subseteq \mathcal{D}(\cH_A), \cF_B \subseteq \mathcal{D}(\cH_B)$, let $\cF_A\otimes \cF_B:= \{\rho_A\otimes \sigma_B: \rho_A\in \cF_A, \sigma_B\in \cF_B\}$. Let $U: \cH_A\rightarrow \cH_A$ be a unitary operator acting on $\cH_A$. 

The definitions below have been adapted from the references \cite{umegaki62,Tomamichel12, TomHay13, li2014, Datta09, JainRS02, GilchristLN05}
For $\rho_A,\sigma_A\in \mathcal{D}(\cH_A)$ such that $\text{supp}(\rho_A) \subset \text{supp}(\sigma_A)$, its relative entropy is given by
\[
\relent{\rho_A}{\sigma_A} := \Tr(\rho_A\log\rho_A) - \Tr(\rho_A\log\sigma_A),
\] 
and is equal to infinity otherwise. Similarly, the relative entropy variance of $\rho_A,\sigma_A\in \mathcal{D}(\cH_A)$ is defined as
\[
\varrelent{\rho}{\sigma} := \Tr(\rho(\log\rho - \log\sigma)^2) - (\relent{\rho}{\sigma})^2,
\]
when $\text{supp}(\rho_A) \subset \text{supp}(\sigma_A)$.

The max-relative entropy of $\rho_A,\sigma_A\in \mathcal{D}(\cH_A)$ such that $\text{supp}(\rho_A) \subset \text{supp}(\sigma_A)$ is 
\[ 
\dmax{\rho_A}{\sigma_A}  :=  \inf \{ \lambda \in \mathbb{R} : 2^{\lambda} \sigma_A \geq \rho_A \}.
\]  
 The smooth max-relative entropy $\dmaxeps{\rho_A}{\sigma_A}{\eps}$ is 
 \[
 \dmaxeps{\rho_A}{\sigma_A}{\eps}  :=  \sup_{\rho'_A\in \ball{\eps}{\rho_A}} \dmax{\rho_A'}{\sigma_A},
 \]
 where the $\varepsilon$-ball, $\ball{\eps}{\rho_A} := \{\rho'_A|~\Pur(\rho_A,\rho'_A) \leq \varepsilon\}$, is defined via the purified distance
\begin{equation}\label{eq_dist}
\Pur(\rho_A,\sigma_A) := \sqrt{1-\F^2(\rho_A,\sigma_A)}.
\end{equation}
In the above, $\F(\rho_A,\sigma_A):=\|\sqrt{\rho_A}\sqrt{\sigma_A}\|_1$ is the fidelity between two states.

We will use the following facts. 
\begin{fact}[Triangle inequality for purified distance, ~\cite{GilchristLN05,Tomamichel12}]
\label{fact:trianglepurified}
For quantum states $\rho_A, \sigma_A, \tau_A$,
$$\Pur(\rho_A,\sigma_A) \leq \Pur(\rho_A,\tau_A)  + \Pur(\tau_A,\sigma_A) . $$ 
\end{fact}

\begin{fact}[Uhlmann's Theorem, \cite{uhlmann76}]
\label{uhlmann}
Let $\rho_A,\sigma_A\in \mathcal{D}(\cH_A)$. Let $\ketbra{\rho}_{AB}\in \mathcal{D}(\cH_{AB})$ be a purification of $\rho_A$. There exists a purification $\ketbra{\theta}_{AB}$ of $\theta_A$ such that,
 $$\F(\ketbra{\theta}_{AB}, \ketbra{\rho}_{AB}) = \F(\rho_A,\sigma_A).$$
\end{fact}

\begin{fact}
\label{closequantumextension}
Let $\rho_A,\sigma_A \in \mathcal{D}(\cH_A)$ and $\rho_{AB}\in \mathcal{D}(\cH_{AB})$ be an extension of $\rho_A$. Then there exists an extension $\sigma_{AB}$ of $\sigma_A$ such that 
$$\F(\rho_{AB},\sigma_{AB}) = \F(\rho_A,\sigma_A).$$
\end{fact}
\begin{proof}
Introduce a register $C$ and consider a purification $\ketbra{\rho}_{ABC}$ of $\rho_{AB}$. There exists a purification $\ketbra{\sigma}_{ABC}$ of $\sigma_A$ such that $\F(\ketbra{\rho}_{ABC}, \ketbra{\sigma}_{ABC}) = \F(\rho_A,\sigma_A)$, due to Uhlmann's Theorem (Fact \ref{uhlmann}). Thus, $\sigma_{AB}$ is the desired extension due to the relation 
$$\F(\rho_A,\sigma_A) \geq \F(\rho_{AB},\sigma_{AB}) \geq \F(\ketbra{\rho}_{ABC}, \ketbra{\sigma}_{ABC}) = \F(\rho_A,\sigma_A).$$ 
\end{proof}

\begin{fact}
\label{closeextension}
Let $\rho_A,\sigma_A \in \mathcal{D}(\cH_A)$ and $\rho_{AB}\in \mathcal{D}(\cH_{AB})$ be a classical-quantum extension of $\rho_A$, that is $\rho_{AB} = \sum_j p_j \rho^j_A\otimes\ketbra{j}_B$, where $\rho_A = \sum_jp_j\rho^j_A$. Then there exists a classical-quantum extension $\sigma_{AB}$ of $\sigma_A$ such that 
$$\F(\rho_{AB},\sigma_{AB}) = \F(\rho_A,\sigma_A),$$ and $\text{supp}(\sigma_B) \subseteq \text{supp}(\rho_B)$.
\end{fact}
\begin{proof}
From Fact \ref{closequantumextension}, there exists an extension $\sigma'_{AB}$ of $\sigma_A$ such that $\F(\rho_{AB}, \sigma'_{AB}) = \F(\rho_A,\sigma_A)$. Now, consider the map $\Lambda : B\rightarrow B$ defined as $\Lambda(\omega_B) = \sum_{j} \ketbra{j}\omega_B\ketbra{j}$. Observe that $\id_A \otimes \Lambda_B(\rho_{AB}) = \rho_{AB}$. Define $\sigma''_{AB}:=\id_A\otimes \Lambda_B(\sigma'_{AB})$. It holds that $\sigma''_{AB}$ is a classical-quantum extension of $\sigma_A$. Moreover, $$\F(\rho_A,\sigma_A) \geq \F(\rho_{AB},\sigma''_{AB}) \geq \F(\rho_{AB}, \sigma'_{AB}) = \F(\rho_A,\sigma_A).$$ Let $\Pi_B$ be the projector onto the support of $\rho_B$. Define the quantum state  $\sigma_{AB} := \frac{\Pi_B\sigma''_{AB}\Pi_B}{\Tr(\sigma''_B\Pi_B)}$. Observe that $\sigma_{AB}$ is also a classical-quantum extension of $\sigma_A$ and $\text{supp}(\sigma_B)\in \text{supp}(\rho_B)$. Then 
\begin{eqnarray*}
\F(\rho_{AB}, \sigma_{AB}) &=& \Tr\left(\sqrt{\sqrt{\rho_{AB}}\sigma_{AB}\sqrt{\rho_{AB}}}\right) \\ &=& \frac{1}{\sqrt{\Tr(\sigma''_B\Pi_B)}}\Tr\left(\sqrt{\sqrt{\rho_{AB}}\Pi_B\sigma''_{AB}\Pi_B\sqrt{\rho_{AB}}}\right) \\ 
&=& \frac{1}{\sqrt{\Tr(\sigma''_B\Pi_B)}}\Tr\left(\sqrt{\sqrt{\rho_{AB}}\sigma''_{AB}\sqrt{\rho_{AB}}}\right) \\ &=& 
\frac{1}{\sqrt{\Tr(\sigma''_B\Pi_B)}}\F(\rho_{AB},\sigma''_{AB}) \\ &\geq& \F(\rho_{AB},\sigma''_{AB}).
\end{eqnarray*}
This completes the  proof.
\end{proof}

\begin{fact}
\label{fact:bipartitedmaxquantum}
Let $\Theta_{AB}$ be a bipartite quantum state. Let $\Pi_B$ be the projector onto the support of $\Theta_B$. Then 
$$\Theta_{AB} \preceq |B|\Theta_A\otimes \Pi_B.$$
\end{fact}

\begin{fact}
\label{fact:bipartitedmax}
Let $\Theta_{AB}$ be a classical quantum state with $B$ as the quantum part. Let $\Pi_B$ be the projector onto the support of $\Theta_B$. Then 
$$\Theta_{AB} \preceq \Theta_A\otimes \Pi_B.$$
\end{fact}

\begin{fact}[\cite{TomHay13, li2014}]
\label{dmaxequi}
Let $\eps\in (0,1)$ and $n$ be an integer. Let $\rho^{\otimes n}, \sigma^{\otimes n}$ be quantum states. Define $\Phi(x) = \int_{-\infty}^x \frac{e^{-t^2/2}}{\sqrt{2\pi}} dt$. It holds that
\begin{equation*}
\dmaxeps{\rho^{\otimes n}}{\sigma^{\otimes n}}{\eps} = n\relent{\rho}{\sigma} + \sqrt{n\varrelent{\rho}{\sigma}} \Phi^{-1}(\eps) + O(\log n) ,
\end{equation*}
\end{fact}

\begin{fact}
\label{gaussianupper}
For the function $\Phi(x) = \int_{-\infty}^x \frac{e^{-t^2/2}}{\sqrt{2\pi}} dt$ and $\eps\leq \frac{1}{2}$, it holds that $|\Phi^{-1}(\eps)| \leq 2\sqrt{\log\frac{1}{2\eps}}$.
\end{fact}
\begin{proof}
We have $$\Phi(-x)=\int_{-\infty}^{-x} \frac{e^{-t^2/2}}{\sqrt{2\pi}} dt = \int_{0}^{\infty} \frac{e^{-(-x-t)^2/2}}{\sqrt{2\pi}} dt \leq e^{-x^2/2} \int_{0}^{\infty} \frac{e^{-(-t)^2/2}}{\sqrt{2\pi}} dt = \frac{1}{2}e^{-x^2/2}.$$ Thus, $\Phi^{-1}(\eps) \geq -2\sqrt{\log\frac{1}{2\eps}}$, which completes the proof.
\end{proof}

\begin{fact}
\label{slowchange}
Let $\rho_1$ be a quantum state and $\{\cE_2,\cE_3, \ldots\}$ be a collection of quantum maps. Define a series of quantum states $\{\rho_2,\rho_3,\ldots \}$ recursively as $\rho_i = \cE_i(\rho_{i-1})$. It holds that $$\Pur(\rho_i,\rho_1) \leq (i-1)\max_i\cdot\{\Pur(\cE_i(\rho_1),\rho_1)\}.$$
\end{fact}
\begin{proof}
Consider 
$$\Pur(\rho_i, \rho_1) = \Pur(\cE_i(\rho_{i-1}),\rho_1) \leq \Pur(\cE_i(\rho_{i-1}),\cE_i(\rho_1)) + \Pur(\cE_i(\rho_1),\rho_1) \leq \Pur(\rho_{i-1},\rho_1) + \Pur(\cE_i(\rho_1),\rho_1).$$ This completes the proof.
\end{proof}

\begin{lemma}[Convex-split lemma \cite{AnshuDJ14}] Let $\eps>0$, $\rho_M$ and $\sigma_M$ be quantum states, $n$ be an integer and $k = \dmaxeps{\rho_M}{\sigma_M}{\eps}$. Consider the following quantum state
\begin{equation*} \label{eq:convexsplit}
 \tau_{M_1M_2\ldots M_n} =  \frac{1}{n}\sum_{j=1}^n \rho_{M_j}\otimes \sigma_{M_1}\otimes \sigma_{M_2}\ldots\otimes\sigma_{M_{j-1}}\otimes\sigma_{M_{j+1}}\ldots\otimes\sigma_{M_n}
\end{equation*}
where $\forall j \in [n]: \rho_{M_j} = \rho_{M}$ and $\sigma_{M_j}=\sigma_M$.  Then,  
$$ \Pur(\tau_{M_1M_2\ldots M_n}, \sigma_{M_1}\otimes\sigma_{M_2}\ldots \otimes \sigma_{M_n}) \leq \eps+ \sqrt{\frac{2^k}{n}}.$$ 
\end{lemma}

\section{Proof of Theorem \ref{thm_oneRER}} \label{sec_oneshot}

\subsection{Achievability}

\begin{theorem}[Achievability]
\label{theo:catalyst}
Fix $\eps,\delta>0$ and $\rho_{M}$. Suppose Assumption \ref{assumeperm} holds. There exists an $(\eps+\delta, k + 2\log\frac{1}{\delta})$-catalytic transformation of $\rho_{M}$ to $\cF$, where $k := \min_{\sigma_{M}\in \cF}\dmaxeps{\rho_{M}}{\sigma_{M}}{\eps}$. 
\end{theorem}
\begin{proof}
Let $n:= \frac{2^{k}}{\delta^2}$ and let $\sigma_{M}\in \cF$ be the quantum state achieving the minimum in the definition of $k$. Let $J$ be a random variable taking values uniformly in $\{1,2,\ldots n\}$. Introduce registers $M_1,M_2,\ldots M_n \equiv M$. Let $\mu_{M_1M_2\ldots M_n} : = \sigma_{M}^{\otimes n}$ be the quantum state in $\cF$ which an experimenter uses as a catalyst.  

The protocol is as follows. Experimenter introduces the quantum state $\mu_{M_1M_2\ldots M_n}\otimes \frac{\id_J}{n}$, which belongs to $\cF$ due to Definition \ref{def:generictransform} (Item $1(b)$). Controlled on the value $j$ in $J$, the experimenter swaps register $M$ with $M_j$. This operation belongs to $\cU$ from Assumption \ref{assumeperm}. Let the resulting global quantum state be $\tau_{MM_1M_2\ldots M_n J}$. We have  
$\tau_{MM_1M_2\ldots M_n} = \sigma_{M}\otimes \tau_{M_1M_2M_2\ldots M_n}$ and
$$\tau_{M_1M_2\ldots M_n} =  \frac{1}{n}\sum_{j=1}^n \rho_{M_j}\otimes \sigma_{M_1}\otimes \sigma_{M_2}\ldots\otimes\sigma_{M_{j-1}}\otimes\sigma_{M_{j+1}}\ldots\otimes\sigma_{M_n}.$$
Using convex-split lemma (Lemma \ref{eq:convexsplit}), this allows us to conclude that
$$\Pur(\tau_{M_1M_2\ldots M_n}, \sigma_{M}^{\otimes n}) \leq \eps + \delta.$$  This completes the proof.
\end{proof}   

\vspace{0.1in} 

\noindent{\bf Remark 1:} The number of qubits used as a catalyst in above task is $\log|M| \cdot \frac{2^k}{\delta^2}$.

\vspace{0.1in}

\subsection{Converse Proof}

In this section, we prove a converse bound for resource transformation, showing near optimality of Theorem \ref{theo:catalyst}. We will separately consider the cases where Assumption \ref{assumeperm} holds and where it does not. 

\begin{theorem}
\label{theo:converse}
Fix an $\eps>0$. For every $(\eps, \log|J|)$- transformation of $\rho_M$ to $\cF$, it holds that 
$$\log |J| \geq \min_{\sigma_M\in \cF}\dmaxeps{\rho_M}{\sigma_M}{\eps},$$ if Assumption \ref{assumeperm} is true.
\end{theorem}
\begin{proof}
Consider any protocol that starts with quantum states $\rho_M \otimes \mu_{EJ}$ and applies the unitary $U := \sum_j U_j\otimes \ketbra{j}_J$ (where $U_j: \cH_E\otimes \cH_M\rightarrow \cH_E\otimes \cH_M$) to obtain $\Theta_{MEJ} = U(\rho_M\otimes \mu_{EJ}) U^{\dagger}$. By Assumption \ref{assumeperm}, $\mu_{EJ}$ is a classical-quantum state with $J$ being the classical register. There exist $\sigma_{ME} \in \cF$ such that $\Pur(\Theta_{ME}, \sigma_{ME}) \leq \eps$.

 From Fact \ref{closeextension}, there exists a classical-quantum extension $\Theta'_{MEJ}$ of $\sigma_{ME}$ (that is, $\Theta'_{ME} = \sigma_{ME}$ and $J$ is the classical register in $\Theta'_{MEJ}$) such that 
\begin{equation}
\label{eq:thetaextension}
\Pur(\Theta_{MEJ}, \Theta'_{MEJ}) = \Pur(\Theta_{ME},\sigma_{ME}) \leq \eps.
\end{equation}
Since $\Theta'_{MEJ}$ is a classical-quantum state, it holds by Fact \ref{fact:bipartitedmax} that 
\begin{equation*}
\Theta'_{MEJ} \preceq \Theta'_{ME}\otimes \id_J = |J|\cdot \sigma_{ME}\otimes \frac{\id_J}{|J|} \implies \dmax{\Theta'_{MEJ}}{\sigma_{ME}\otimes \frac{\id_J}{|J|}} \leq \log|J|.
\end{equation*}
 From this, we conclude that 
\begin{equation}
\label{eq:lowerboundJ}
\dmax{\Tr_{EJ}\left(U^{\dagger}\Theta'_{MEJ}U\right)}{\Tr_{EJ}\left(U^{\dagger}\left(\sigma_{ME}\otimes \frac{\id_J}{|J|}\right) U\right)} \leq \log|J|.
\end{equation}
From Equation (\ref{eq:thetaextension}),  
\begin{equation}
\label{eq:ballofrho}
\Pur\left(\Tr_{EJ}\left(U^{\dagger}\Theta'_{MEJ}U\right), \rho_M\right) \leq \Pur(U^{\dagger}\Theta'_{MEJ}U, \rho_M\otimes\mu_{EJ}) = \Pur(\Theta'_{MEJ}, \Theta_{MEJ}) \leq \eps.
\end{equation}
From Definition \ref{def:generictransform} (Items $1 (a), 3$), $$\Tr_{J}\left(U^{\dagger}\left(\sigma_{ME}\otimes \frac{\id_J}{|J|} \right)U\right) = \frac{1}{|J|}\sum_j U^{\dagger}_j(\sigma_{ME}) U_j \in \cF.$$ Thus, again from Definition \ref{def:generictransform} (Item $1 (c)$), 
\begin{equation}
\label{eq:stateinC2}
\Tr_{EJ}\left(U\left(\sigma_M\otimes \mu'_E\otimes \frac{\id_J}{|J|} \right)U^{\dagger}\right) \in \cF.
\end{equation}
Combining Equations (\ref{eq:lowerboundJ}), (\ref{eq:ballofrho}) and (\ref{eq:stateinC2}), we obtain 
$$ \min_{\sigma'_M \in \cF}\dmaxeps{\rho_M}{\sigma'_M}{\eps} \leq \log|J|.$$ This completes the proof.

\end{proof}

Now we consider the converse bound without Assumption \ref{assumeperm}.

\begin{theorem}
\label{converse_general}
Fix an $\eps>0$. For every $(\eps, \log|J|)$- transformation of $\rho_M$ to $\cF$, it holds that 
$$\log |J| \geq \frac{1}{2}\min_{\sigma_M\in \cF}\dmaxeps{\rho_M}{\sigma_M}{\eps}.$$
\end{theorem}
\begin{proof}
Consider any protocol that starts with quantum states $\rho_M \otimes \mu_{EJ}$ and applies the unitary $U \in \cU$ to obtain $\Theta_{MEJ} = U(\rho_M\otimes \mu_{EJ}) U^{\dagger}$. There exist $\sigma_{ME} \in \cF$ such that $\Pur(\Theta_{ME}, \sigma_{ME}) \leq \eps$.

From Fact \ref{closequantumextension}, there exists an extension $\Theta'_{MEJ}$ of $\sigma_{ME}$ (that is, $\Theta'_{ME} = \sigma_{ME}$) such that 
\begin{equation}
\label{eq:thetaextensiongen}
\Pur(\Theta_{MEJ}, \Theta'_{MEJ}) = \Pur(\Theta_{ME},\sigma_{ME}) \leq \eps.
\end{equation}
It holds by Fact \ref{fact:bipartitedmaxquantum} that 
\begin{equation*}
\Theta'_{MEJ} \preceq |J|\Theta'_{ME}\otimes \id_{J} = |J|^2\cdot \sigma_{ME}\otimes \frac{\id_J}{|J|} \implies \dmax{\Theta'_{MEJ}}{\sigma_{ME}\otimes \frac{\id_{J}}{|J|}} \leq 2\log|J|.
\end{equation*}
 From this, we conclude that 
\begin{equation}
\label{eq:lowerboundJgen}
\dmax{\Tr_{EJ}\left(U^{\dagger}\Theta'_{MEJ}U\right)}{\Tr_{EJ}\left(U^{\dagger}\left(\sigma_{ME}\otimes \frac{\id_{J}}{|J|}\right) U\right)} \leq 2\log|J|.
\end{equation}
From Equation (\ref{eq:thetaextensiongen}),  
\begin{equation}
\label{eq:ballofrhogen}
\Pur\left(\Tr_{EJ}\left(U^{\dagger}\Theta'_{MEJ}U\right), \rho_M\right) \leq \Pur(U^{\dagger}\Theta'_{MEJ}U, \rho_M\otimes\mu_{EJ}) = \Pur(\Theta'_{MEJ}, \Theta_{MEJ}) \leq \eps.
\end{equation}
From Definition \ref{def:generictransform} (Items $1 (b), 1 (d), 3$), $$\Tr_{J}\left(U^{\dagger}\left(\sigma_{ME}\otimes \frac{\id_{J}}{|J|} \right)U\right) \in \cF.$$ Thus, again from Definition \ref{def:generictransform} (Item $1 (c)$), 
\begin{equation}
\label{eq:stateinC2gen}
\Tr_{EJ}\left(U\left(\sigma_M\otimes \mu'_E\otimes \frac{\id_{J}}{|J|} \right)U^{\dagger}\right) \in \cF.
\end{equation}
Combining Equations (\ref{eq:lowerboundJgen}), (\ref{eq:ballofrhogen}) and (\ref{eq:stateinC2gen}), we obtain 
$$ \min_{\sigma'_M \in \cF}\dmaxeps{\rho_M}{\sigma'_M}{\eps} \leq 2\log|J|.$$ This completes the proof.

\end{proof}

\section{Proof of Theorem \ref{thm_multiRER}.}\label{sec_multioneshot}

\subsection*{Achievability}

The proof of achievability follows along lines similar to proof of Theorem \ref{theo:catalyst}. We state it as a corollary.
\begin{cor}
\label{cor:multipartite}
Fix $\eps,\delta>0$ and $\rho_{M^1M^2\ldots M^t}$. Suppose Assumption \ref{assumepermmul} holds. Then there exists an $(\eps+\delta, k + 2\log\frac{1}{\delta}, t)$-catalytic transformation of $\rho_{M^1M^2\ldots M^t}$ to $\cF$, where $$k := \min_{\sigma_{M^1M^2\ldots M^t}\in \cF}\dmaxeps{\rho_{M^1M^2\ldots M^t}}{\sigma_{M^1M^2\ldots M^t}}{\eps}.$$ 
\end{cor}
\begin{proof}
Let $n:=\frac{2^k}{\delta^2}$ and $\sigma_{M^1M^2\ldots M^t}\in \cF$ be the quantum state achieving the minimum in the definition of $k$. The experimenters possess the state $\id_{n,t}$ in registers $J_1J_2\ldots J_t$, such that $|J_i|=n$ for all $i$. Introduce registers $M^i_1,M^i_2,\ldots M^i_n\equiv M^i$. They use the quantum state $$\mu_{M^1_1M^2_1\ldots M^t_1, M^1_2M^2_2\ldots M^t_2,\ldots M^1_nM^2_n\ldots M^t_n}:=\sigma^{\otimes n}_{M^1M^2\ldots M^t}$$ as a catalyst.  The protocol closely follows the protocol given in the proof of Theorem \ref{theo:catalyst}, where the party $i$ swaps the registers $M^i$ and $M^i_j$ controlled on the value $j$ in register $J_i$. It can be verified from Definition \ref{def:multigenerictransform} and Assumption \ref{assumepermmul} that all these operations belong to $\cF$ and $\cU$. The proof of completeness follows similarly using convex split lemma (Lemma \ref{eq:convexsplit}). 
\end{proof}

\subsection*{Converse}

Using argument similar to that of Theorem \ref{theo:converse}, we have the following corollary.
\begin{cor}
\label{cor:converse}
Fix an $\eps>0$. For every $(\eps, \log|J|, t)$- transformation of $\rho_{M^1M^2\ldots M^t}$ to $\cF$, it holds that 
$$\log |J| \geq \min_{\sigma_{M^1M^2\ldots M^t}\in \cF}\dmaxeps{\rho_{M^1M^2\ldots M^t}}{\sigma_{M^1M^2\ldots M^t}}{\eps},$$ if Assumption \ref{assumepermmul} is true.
\end{cor}
\begin{proof}

Consider any protocol that starts with quantum states $\rho_{M^1M^2\ldots M^t} \otimes \mu_{EJ_1J_2\ldots J_t}$ and applies the unitary $U := \sum_k U_k\otimes \ketbra{k}_{J_1}\otimes \ketbra{k}_{J_2}\otimes \ldots \ketbra{k}_{J_t}$ (where $U_k: \cH_{M^1M^2\ldots M^t}\otimes \cH_E\rightarrow \cH_{M^1M^2\ldots M^t}\otimes \cH_E$) to obtain $\Theta_{M^1M^2\ldots M^tEJ_1J_2\ldots J_t} = U(\rho_{M^1M^2\ldots M^t}\otimes \mu_{EJ_1J_2\ldots J_t}) U^{\dagger}$. By Assumption \ref{assumepermmul}, $\mu_{EJ_1J_2\ldots J_t}$ is a classical-quantum state with registers $J_1J_2\ldots J_t$ being classical and $\mu_{J_1J_2\ldots J_t}$ belonging to $\text{supp}(\id_{J,t})$. Furthermore, there exists a quantum state $\sigma_{M^1M^2\ldots M^tE} \in \cF$ such that $\Pur(\Theta_{M^1M^2\ldots M^tE}, \sigma_{M^1M^2\ldots M^tE}) \leq \eps$.

 From Fact \ref{closeextension}, there exists a classical-quantum extension $\Theta'_{M^1M^2\ldots M^tEJ_1J_2\ldots J_t}$ of $\sigma_{M^1M^2\ldots M^tE}$ (that is, $\Theta'_{M^1M^2\ldots M^tE} = \sigma_{M^1M^2\ldots M^tE}$ and $J_1J_2\ldots J_t$ is the classical register in $\Theta'_{M^1M^2\ldots M^tEJ_1J_2\ldots J_t}$) such that 
\begin{equation}
\label{eq:thetaextensionmul}
\Pur(\Theta_{M^1M^2\ldots M^tEJ_1J_2\ldots J_t}, \Theta'_{M^1M^2\ldots M^tEJ_1J_2\ldots J_t}) = \Pur(\Theta_{M^1M^2\ldots M^tE},\sigma_{M^1M^2\ldots M^tE}) \leq \eps.
\end{equation}
Moreover, $\text{supp}(\Theta'_{J_1J_2\ldots J_t}) \subseteq \text{supp}(\id_{J,t})$. Thus, by Fact \ref{fact:bipartitedmax}  
\begin{equation*}
\Theta'_{M^1M^2\ldots M^tEJ_1J_2\ldots J_t} \preceq |J|\Theta'_{M^1M^2\ldots M^tE}\otimes \id_{|J|,t} 
\end{equation*}
which implies
$$ \dmax{\Theta'_{M^1M^2\ldots M^tEJ_1J_2\ldots J_t}}{\sigma_{M^1M^2\ldots M^tE}\otimes \id_{|J|,t}} \leq \log|J|.$$
 From this, we conclude that 
\begin{equation}
\label{eq:lowerboundJmul}
\dmax{\Tr_{EJ_1J_2\ldots J_t}\left(U^{\dagger}\Theta'_{M^1M^2\ldots M^tEJ_1J_2\ldots J_t}U\right)}{\Tr_{EJ_1J_2\ldots J_t}\left(U^{\dagger}\left(\sigma_{M^1M^2\ldots M^tE}\otimes \id_{|J|,t}\right) U\right)} \leq \log|J|.
\end{equation}
From Equation (\ref{eq:thetaextensionmul}),  
\begin{eqnarray}
\label{eq:ballofrhomul}
&&\Pur\left(\Tr_{EJ_1J_2\ldots J_t}\left(U^{\dagger}\Theta'_{M^1M^2\ldots M^tEJ_1J_2\ldots J_t}U\right), \rho_{M^1M^2\ldots M^t}\right) \nonumber\\ &&\leq \Pur(U^{\dagger}\Theta'_{M^1M^2\ldots M^tEJ_1J_2\ldots J_t}U, \rho_{M^1M^2\ldots M^t}\otimes\mu_{EJ_1J_2\ldots J_t}) \nonumber\\ &&= \Pur(\Theta'_{M^1M^2\ldots M^tEJ_1J_2\ldots J_t}, \Theta_{M^1M^2\ldots M^tEJ_1J_2\ldots J_t}) \leq \eps.
\end{eqnarray}
From Definition \ref{def:multigenerictransform} (Items $1 (a), 3$), $$\Tr_{J_1J_2\ldots J_t}\left(U^{\dagger}\left(\sigma_{M^1M^2\ldots M^tE}\otimes \id_{|J|,t} \right)U\right) = \frac{1}{|J|}\sum_k U^{\dagger}_k(\sigma_{M^1M^2\ldots M^tE}) U_k \in \cF.$$ Thus, again from Definition \ref{def:multigenerictransform} (Item $1 (c)$), 
\begin{equation}
\label{eq:stateinC2mul}
\Tr_{EJ_1J_2\ldots J_t}\left(U\left(\sigma_{M^1M^2\ldots M^t}\otimes \mu'_E\otimes \id_{|J|,t} \right)U^{\dagger}\right) \in \cF.
\end{equation}
Combining Equations (\ref{eq:lowerboundJmul}), (\ref{eq:ballofrhomul}) and (\ref{eq:stateinC2mul}), we obtain 
$$ \min_{\sigma'_{M^1M^2\ldots M^t} \in \cF}\dmaxeps{\rho_{M^1M^2\ldots M^t}}{\sigma'_{M^1M^2\ldots M^t}}{\eps} \leq \log|J|.$$ This completes the proof.

\end{proof}

\section{Proof of Theorem \ref{coro_RER}}\label{sec_asymp}

First we show that the number of qubits of catalysts scales only polynomially in the number of copies of the quantum state to be transformed.

\begin{theorem}
\label{theo:asymptotictransform}
Suppose Assumption \ref{assumeperm} holds. Let $\eps,\gamma>0$ be such that $\gamma^2\leq \eps$ and $m$ be a sufficiently large integer. For every $\rho_{M}$, there exists a $(\eps, m(R+\gamma))$-catalytic transformation of $\rho_{M}^{\otimes m}$ to $\cF$ , such that the number of qubits of catalyst used upper bounded by 
$$|M|\cdot \frac{2\log m \cdot V}{\gamma^2} \cdot m^{\frac{2\cdot V \cdot (R +\gamma)}{\gamma^2}},$$ where $R := \min_{\sigma_{M}\in \cF} \relent{\rho_{M}}{\sigma_{M}}$ and $V:= \varrelent{\rho_{M}}{\sigma^*_{M}}$, where $\sigma^*_M$ achieves the minimum in the definition of $R$.
\end{theorem}
\begin{proof}
Fix a quantum state $\sigma_{M}\in \cF$ and let $$\ellse:= \frac{2\log m \cdot \varrelent{\rho_{M}}{\sigma_{M}}}{\gamma^2} , k:=\dmaxeps{\rho_{M}^{\otimes \ellse}}{\sigma_{M}^{\otimes \ellse}}{\frac{\eps\ellse}{2m}} + 2\log\frac{2m}{\eps\ellse}.$$ Let $J$ be a random variable taking values uniformly in $\{1,2,\ldots 2^k\}$. Introduce registers $M'_1,M'_2,\ldots M'_{2^k}$ such that each $M'_i$ is equivalent to $\ellse$ copies of $M$. Experimenter introduces the quantum state $$\sigma_M^{\otimes \ellse\cdot 2^k} : =(\sigma^{\otimes \ellse})_{M'_1}\otimes \ldots (\sigma^{\otimes \ellse})_{M'_{2^k}}.$$ She generates $\frac{m}{\ellse}$ copies $J_1,J_2,\ldots J_{\frac{m}{\ellse}}$ of the random variable $J$.  Let $\cP$ be the $(\frac{\eps\ellse}{m}, k)$-catalytic transformation protocol for $\rho_{M}^{\otimes \ellse}$, as guaranteed by Theorem \ref{theo:catalyst}. The desired protocol $\cP'$ is as follows, where the quantum state $\rho_{M}^{\otimes m}$ is divided into $\frac{m}{\ellse}$ blocks as $$\rho_{M}^{\otimes m} = (\rho^{\otimes \ellse})_{M_1}\otimes (\rho^{\otimes \ellse})_{M_2}\otimes \ldots (\rho^{\otimes \ellse})_{M_{\frac{m}{\ellse}}}.$$ Here, $M_i$ is equivalent to $\ellse$ copies of $ M$ for all $i\in \{1,2, \ldots \frac{m}{\ellse}\}$.
\begin{enumerate}
\item Set $i=1$.
\item Experimenter introduces the random variable $J_i$ and runs the protocol $\cP$ for the $i$-th block with quantum state $(\rho^{\otimes \ellse})_{M_i}$.  
\item She throws away the random variable $J_i$.
\item Set $i\leftarrow i+1$. Go to step $2$.
\end{enumerate}

Let $\Theta$ be the quantum state on registers $M_1,M_2,\ldots M_{\frac{m}{\ellse}}, M'_1,M'_2,\ldots M'_{2^k}$ after the protocol ends. For iteration $i$, let $\cP'_i$ be the overall quantum map applied on these registers after step $3$ finishes. Observe that $\cP'_i$ acts only on registers $M_i, M'_1, M'_2, \ldots M'_{2^k}$. From Theorem \ref{theo:catalyst}, we have that 
$$\Pur\left(\cP'_i\left((\rho^{\otimes \ellse})_{M_i}\otimes \sigma_M^{\otimes \ellse\cdot 2^k} \right), (\sigma^{\otimes \ellse})_{M_i}\otimes \sigma_M^{\otimes \ellse\cdot 2^k}  \right) \leq \frac{\eps\ellse}{m}.$$
Thus,  from Fact \ref{slowchange}, $$\Pur(\Theta, (\sigma^{\otimes \ellse})_{M_1}\otimes \ldots (\sigma^{\otimes \ellse})_{M_{\frac{m}{\ellse}}}\otimes (\sigma^{\otimes \ellse})_{M'_1}\otimes \ldots (\sigma^{\otimes \ellse})_{M'_{2^k}}) \leq (\frac{m}{\ellse}-1)\frac{\eps\ellse}{m} \leq \eps.$$ Thus, the protocol achieves an $(\eps,\frac{m k}{\ellse})$-catalytic transformation of $\rho_{M}^{\otimes m}$. The theorem now follows from the following upper bound on $k$.
\begin{eqnarray*}
k &=& \dmaxeps{\rho_{M}^{\otimes \ellse}}{\sigma_{M}^{\otimes \ellse}}{\frac{\eps\ellse}{2m}} + 2\log\frac{2m}{\eps\ellse} \\
&\leq& \ellse \cdot \relent{\rho_{M}}{\sigma_{M}} + \sqrt{\ellse\cdot \varrelent{\rho_{M}}{\sigma_{M}}}\Phi^{-1}\left(\frac{\eps\ellse}{2m}\right) + O(\log \frac{2m}{\eps\ellse}) \\ &\leq&  \ellse \cdot \relent{\rho_{M}}{\sigma_{M}} + \sqrt{\ellse \log\frac{2m}{\eps\ellse}\cdot \varrelent{\rho_{M}}{\sigma_{M}}} + O(\log \frac{m}{\eps\ellse})\\ & = & \ellse \left( \relent{\rho_{M}}{\sigma_{M}} + \sqrt{\frac{1}{\ellse} \log\frac{2m}{\eps\ellse}\cdot \varrelent{\rho_{M}}{\sigma_{M}}} + O(\frac{\log \frac{m}{\eps\ellse}}{\ellse})\right) \\ &\leq & \ellse \left(\relent{\rho_{M}}{\sigma_{M}} + \gamma\right),
\end{eqnarray*}
where we have used Facts \ref{dmaxequi} and \ref{gaussianupper}, the choice of $\ellse$ and the relation $\gamma^2\leq \eps$.

\end{proof}

Now we provide optimal asymptotic i.i.d. analysis of Task \ref{catalytictransform}. For this we require the following result, essentially proved in \cite{DHorodecki99}, and restated in our setting.
\begin{lemma}[\cite{DHorodecki99}]
\label{relentresourcecont}
Let $\rho,\rho' \in \mathcal{D}(\cH_M)$ be quantum states on register $M$ with $\|\rho-\rho'\|_1:= \eps \leq \frac{1}{3}$. Let $\cF\subseteq \mathcal{D}(\cH_M)$ be a convex set. Then it holds that 
$$|\inf_{\sigma\in \cF}\relent{\rho}{\sigma} - \inf_{\sigma'\in \cF}\relent{\rho'}{\sigma'}| \leq \eps \left(\log M+ \inf_{\tau\in \cF}\|\log\tau\|_\infty\right) + \eps\log\frac{1}{\eps} + 4\eps.$$
\end{lemma}

Using above lemma, we have the following theorem.

\begin{theorem}
\label{theo:asymptoticvalue}
Suppose Assumption \ref{assumeperm} holds and $C(\cF) < \infty$. Then for every $\rho_M$, the asymptotic randomness rate of catalytic transformation of $\rho_M$ is equal to 
$$\lim_{n\rightarrow \infty}\frac{1}{n}\min_{\sigma\in \cF}\relent{\rho_M^{\otimes n}}{\sigma}.$$ 
\end{theorem}
\begin{proof}
The upper bound follows by applying Theorem \ref{theo:asymptotictransform} to the quantum state $\rho^{\otimes n}_M$. For the lower bound, we appeal to the converse in Theorem \ref{theo:converse} to conclude that for the $(\eps, nR)$-catalytic transformation of $\rho_M^{\otimes n}$ to $\cF$, we require
$$R \geq \frac{1}{n}\min_{\sigma\in \cF}\dmaxeps{\rho_M^{\otimes n}}{\sigma}{\eps} \geq \frac{1}{n}\min_{\sigma\in \cF}\relenteps{\rho_M^{\otimes n}}{\sigma}{\eps}.$$ Now, observe that 
$$\min_{\sigma\in \cF}\relenteps{\rho_M^{\otimes n}}{\sigma}{\eps} = \min_{\sigma\in \cF}\min_{\rho'\in \ball{\eps}{\rho^{\otimes n}_M}}\relent{\rho'}{\sigma} = \min_{\rho'\in \ball{\eps}{\rho^{\otimes n}_M}}\min_{\sigma\in \cF}\relent{\rho'}{\sigma}.$$ Let $\rho' \in \ball{\eps}{\rho^{\otimes n}_M}$ be the quantum state that achieves the minimum above. From Lemma \ref{relentresourcecont}, it holds that $$\frac{1}{n}|\relent{\rho'}{\sigma}-\relent{\rho^{\otimes n}}{\sigma}| \leq \eps C(\cF) + \eps\log M +\frac{5}{n}\eps\log\frac{1}{\eps}.$$ Thus, 
$$R\geq \frac{1}{n}\min_{\sigma\in \cF}\relent{\rho_M^{\otimes n}}{\sigma} - \eps(\log|M| + C(\cF))) - \frac{5}{n}\eps\log\frac{1}{\eps}.$$ Letting $n\rightarrow \infty$ and $\eps\rightarrow 0$, we conclude the proof.

\end{proof}

\section{Conclusion}\label{sec_cld}
In this work we quantity the amount of `resource' contained in a given state, in terms of regularized relative entropy of resource, in a general resource theoretic framework that allows to use free resource as a catalyst. This then yields a new operational interpretation to this entropic quantity, which has a nice geometric interpretation. Our general catalytic resource framework, while having some restriction on the free resource and free operations, is still general enough to include several well-developed resource theories in the literature, including resource theory of entanglement, resource theory of coherence, resource theory of non-uniformity,  classical resource theory of randomness extraction, etc. 

Our main result follows by first providing a one-shot bound that characterizes the amount of randomness required in order to erase the resource contained in one copy of a given state. The one-shot bound is given by smooth max-relative entropy. This also gives the smooth max-relative entropy a new operational meaning in the resource theoretic framework. To the best of our knowledge, this is the first one-shot result that works for a general resource theoretic framework. This result is obtained via the convex-split lemma. 

Answering the amount of useful resource possessed by a given state is just a  first step toward a general resource theory. Understanding what other information-processing tasks could be done in this general framework is thus a meaningful and fruitful open question.

\subsection*{Acknowledgement}
We thank Mario Berta and Christian Majenz for helpful discussions. Part of this work was done when the authors were visiting Centrum Wiskunde $\&$ Informatica, Amsterdam and during the Institute for Mathematical Sciences workshop `Beyond I.I.D. in Information Theory'. This work is supported by the Singapore Ministry of Education and the National Research Foundation, also through the Tier 3 Grant “Random numbers from quantum processes” MOE2012-T3-1-009.

\bibliographystyle{ieeetr}
\bibliography{References}    
\end{document}